\documentclass[graybox]{svmult}

\usepackage{mathptmx}       
\usepackage{helvet}         
\usepackage{courier}        
\usepackage{type1cm}        
\usepackage{makeidx}         
\usepackage{graphicx}        
\usepackage{multicol}        
\usepackage[bottom]{footmisc}

\usepackage{color}
\usepackage{amsfonts, amsmath, amssymb}

\definecolor{RedClr}{rgb}{1,0,0}
\definecolor{BlueClr}{rgb}{0,0,1}
\definecolor{TextColor}{rgb}{0,0,0.5}
\definecolor{Violet}{rgb}{0.5,0,1}
\definecolor{Bordeaux}{rgb}{1,0.3,0.4}

\usepackage{xargs}
\newcommandx{\MW}[2][1=]{\todo[linecolor=red,backgroundcolor=red!25,bordercolor=red,#1]{#2}}
\newcommandx{\ES}[2][1=]{\todo[linecolor=blue,backgroundcolor=blue!25,bordercolor=blue,#1]{#2}}
\newcommandx{\JJ}[2][1=]{\todo[linecolor=violet,backgroundcolor=violet!25,bordercolor=blue,#1]{#2}}

\makeindex             

\begin{document}

\title*{Forman's Ricci curvature - From networks to hypernetworks}
\author{Emil Saucan$^\ast$ and Melanie Weber$^\ast$} 
\institute{Emil Saucan$^\ast$ \at ORT Braude College of Technology, Karmiel 2161002, Israel \email{semil@braude.ac.il}
\and Melanie Weber$^\ast$ \at Princeton University, Princeton NJ 08540, United States \email{mw25@math.princeton.edu} 
\and $^\ast$: co-first authors}

%
%
\maketitle


\abstract{Networks and their higher order generalizations, such as hypernetworks or multiplex networks are ever more popular models in the applied sciences. However, methods developed for the study of their structural properties go little beyond the common name and the heavy reliance of combinatorial tools. We show that, in fact, a geometric unifying approach is possible, by viewing them as polyhedral complexes endowed with a simple, yet, the powerful notion of curvature -- the Forman Ricci curvature. We systematically explore some aspects related to the modeling of weighted and directed hypernetworks and present expressive and natural choices involved in their definitions. A benefit of this approach is a simple method of structure-preserving embedding of hypernetworks in Euclidean $N$-space. Furthermore, we introduce a simple and efficient manner of computing the well established Ollivier-Ricci curvature of a hypernetwork. 
}

\section{Introduction}\label{sec:intro}
Networks are popular means for modeling (pairwise) relationships between elements in a complex system. While the simplicity of such models allows for efficient representation and computational analysis, its ability to represent complex relations is limited. In the context of applications in the social and biological sciences, we often encounter relationships between groups of elements or systems with different types of elements that can participate in a specified set of relationship types. As an effective model for such higher-order relationships, hypernetworks have recently attracted increasing interest in the network science community. However, most network analysis tools are targeted to classic networks and do not readily transfer to the case of hypernetworks.

In the present work, we introduce geometric tools for the analysis of hypernetworks. We suggest a parametrization that allows for the application of curvature-based tools~\cite{WSJ,SSWJ} to hypernetworks, allowing for an efficient analysis of their structure and intrinsic geometry. The introduced methods are based on a discrete notion of Ricci curvature~\cite{Fo} whose applicability for network analysis and data science applications has recently been studied by the authors and collaborators for a wide range of network data~\cite{WSJ,WSJ2,SSWJ,SVJSS,FR-brains}.

We suggest representing hypernetworks as weighted cell complexes, more specifically as simplicial or polyhedral complexes. In this parametrization, faces encode relationships between groups of nodes; for example, triangles may represent relationships between triples of nodes. Note, that we distinguish between \emph{pairwise} relationships, as encoded in the classic networks model, and \emph{higher-order} relationships. For instance, tetrahedra represent higher-order relationships between quadruples of nodes, whereas quadrangles denote pairwise relationships between four nodes. For practical purposes, simplicial complexes are the most suitable parametrization, as they are easy to analyze and topologically robust~\cite{CB,R++}. However, this simplified model falls short of capturing relationships among groups of more than three vertices. In~\cite{WSJ} we demonstrate empirically, that such relationships occur frequently in real-world data sets, creating a need for higher-order models that incorporate these structures. We will, therefore, introduce a parametrization that represents hypernetworks as polyhedral complexes instead. 

We have touched on these ideas and their significance for modeling complex systems in~\cite{WSJ} and~\cite{SJ} and discussed some of its geometrical and topological implications. Here, we will extend these ideas and discuss (potential) applications. 

We would like to thank J{\"u}rgen Jost and Areejit Samal for helpful discussions.

\section{Hypernetworks as Polyhedral Complexes}\label{sec:param}
In classic graphs, an edge connects exactly two nodes, i.e. it represents a \emph{pairwise} relation. Here, we consider the more general notion of \emph{hypernetworks} that allows for modeling relations between groups of nodes.

\subsection{Parametrizing hypernetworks as polyhedral complexes}
\begin{definition}[Hypernetworks]
We define a \emph{hypernetwork} as a {\it hypergraph} $\mathcal{H} = (\mathcal{V},\mathcal{E})$ with the {\it node set} $V$ partitioned into {\it type sets}, i.e. $\mathcal{V} = (V_1,\ldots,V_p)$, and hyperedges $\left[ \left( v_1, \dots, v_n	\right), \left(	v_1, \dots, v_m	\right) \right] \in E$ connecting groups of nodes.
\end{definition}
Since for $m=n=1$, one returns to the classical case, it is natural to assign to each pair of nodes appertaining to a hyperedge a segment/edge. This observation gives rise to a natural parametrization of a hyperedge with $m$/ $n$-dimensional simplices, forming a simplicial complex $S$. For instance, we can parametrize parent-child relationships as a directed complex (see Figure \ref{fig:HyperNet}). Note that the resulting complex is simplicial, if and only if we only allow connections between nodes of the same type. Otherwise, the resulting complex is polyhedral. Thus, hypernetworks can be parametrized as {\it bi-partite} simplicial or polyhedral complexes. 
\begin{figure}[h] \label{fig:HyperNet}
	\begin{center}
		\includegraphics[width=.75\columnwidth]{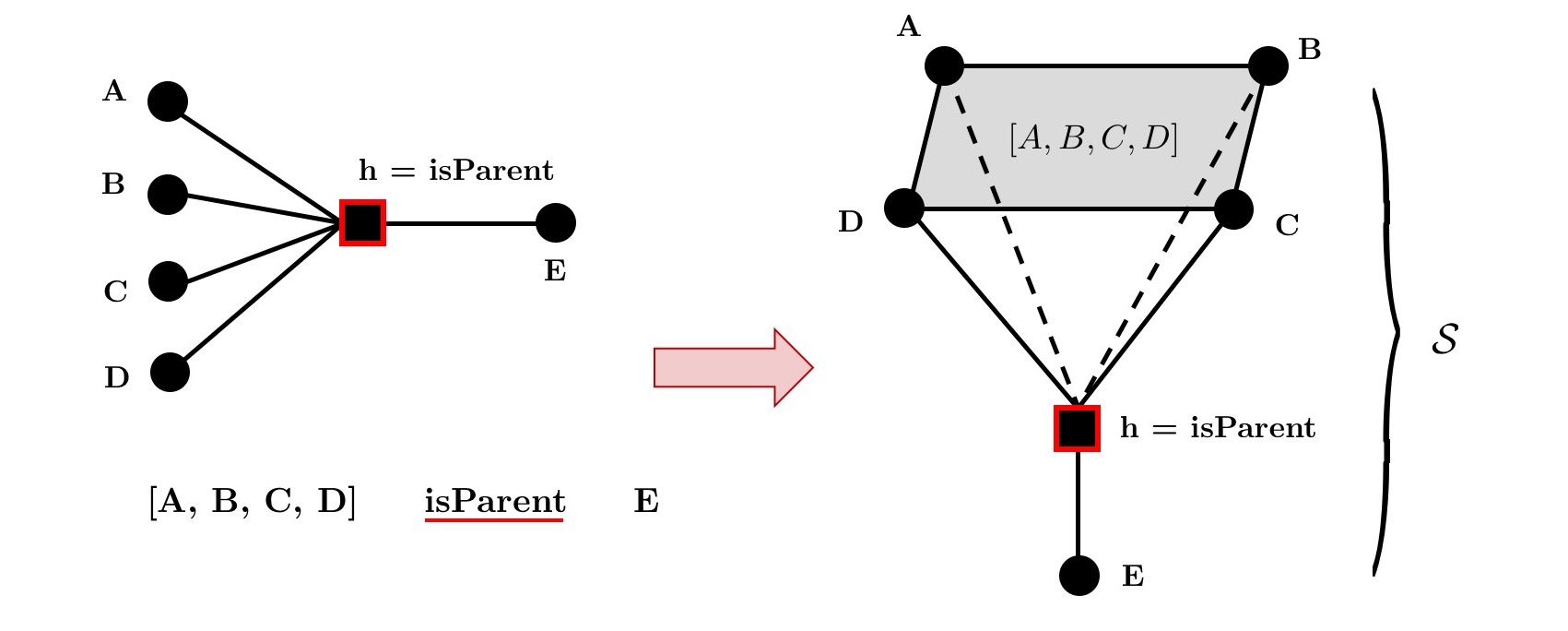}
	\end{center}
	\caption{Parametrization of hypernetworks as simplicial complexes. A child-parent relationship is modeled as a hyperedge $h$: $A$, $B$, $C$ and $D$ are parents of $E$ (left). We represent the co-parent relationship as a face $\left[	A , B, C, D \right]$ and represent $h$ as a simplicial complex $\mathcal{S}=\left[	A, B, C, D, E\right]$ (right). }
\end{figure}
A closely related model for hypernetworks is that of bipartite, directed graphs, albeit with the restriction that edges only connect nodes of the same type. This model does not necessarily capture higher-dimensional (geometric) information. In the simplicial/ polyhedral complex model, such information can be encoded in the weights of faces and higher order simplices. We will discuss such parametrizations in section~\ref{sec:Weights}. Related work on such representations includes the {\it clique expansion} and the {\it star expansion} -- see, e.g. \cite{Tu+}. 
An additional advantage of this model arises from the {\it structure-preserving} low-dimensional embedding of a simplicial complex with {\it indecomposable edges} (see \cite{Tu+}), that is not readily available for bipartite graphs. To understand this observation, let us introduce the following notation: Let ${\rm dim}(\mathcal{V}) = \max\{{\rm dim}(V) \,|\, V \in \mathcal{V}\}$, ${\rm dim}(\mathcal{E}) = \max\{{\rm dim}(E) \,|\, E \in \mathcal{E}\}$ denote the maximal dimension of the hyperedges, and define the {\it dimension} of the hypernetwork $\mathcal{H}$ as ${\rm dim}(\mathcal{H}) = \max\{{\rm dim}(\mathcal{V}), {\rm dim}(\mathcal{E})\}$. 
We can now state the following
\begin{theorem}
	Any hypernetwork has a structure-preserving embedding as an Euclidean complex in $\mathbb{R}^{2{\rm dim}(H)+1}$. 
\end{theorem}
\begin{proof}
By a classical result in $PL$-Topology (\cite{Fl,vK}), any $n$-dimensional polyhedral complex admits a (combinatorial) embedding in $\mathbb{R}^{2n+1}$. Therefore, the hypernetwork $H$, viewed as a complex, admits an embedding in the lowest dimension in which all of its constituting components can be embedded, namely ${\rm dim}(H)$. $\Box$
\end{proof}


\subsection{Choice of weights}
\label{sec:Weights}
%
The problem of choosing appropriate weights of the variously dimensional faces has two different, yet complementary aspects. The first is empirical: The choice depends on the specific application to be modeled, and it is the resort of the practitioners to decide on the relevant features and information to be encoded. The second is the geometric aspect that studies the choice of weights from a purely mathematical viewpoint. In this section, we discuss how to produce expressive weights for edges and higher dimensional faces from geometric and empirical information.

For \emph{combinatorial weights}, i.e for the standard allocation of weight one to each edge, one also canonically assigns weight equal to one to each face, in any dimension. We do not concentrate on this model here, since we consider it overly simplistic, given that it captures only the basic underlying topology and very little of the geometry of a (hyper)network. 

For choosing edge weights in the \emph{general case}, one can either settle for generalized weights in Forman's framework~\cite{Fo} or  metricize the graph. More expressive metrics than the combinatorial one include the consideration of both edge weights viewed as abstractions of lengths and node weights as measures concentrated at the nodes. A number of choices are available, mainly the {\it path degree metric} \cite{DK} and the well known Wasserstein distance $W_1$ (see e.g. \cite{Vi}). 

The case of higher dimensional faces is more delicate, as uniqueness is usually not guaranteed. Our choice of weights introduced below is purely geometric and extends to any dimension without imposing thresholds. Related approaches that rely on thresholding (see, e.g.,~\cite{R++}) introduce additional degrees of freedom in the choice of the threshold value.
	
    For 2-dimensional faces in simplicial complexes (triangles), the obvious geometric choice is that of the area. Moreover, since such complexes are usually envisioned as {\it piecewise flat}, i.e. composed of Euclidean triangles, etc., in a metrized graph, the area (weight) of any face $T$ can be easily  expressed as a function of the weights (lengths) of the edges, via the {\it Heron Formula}:
    \begin{equation} \label{eq:Heron}
    w_2(T) = \sqrt{p(p-w_1)(p-w_2)(p-w_3)}, 
    \end{equation}
    where $p = (w_1+w_2+w_3)/2$, and where $w_1,w_2,w_3$ represent the weights (lengths) of the edges.
    %
    %
    However, {\emph piecewise} linear ($PL$) does not necessarily mean piecewise {\it flat} and, indeed network embeddings in Hyperbolic space have been considered \cite{BR}, \cite{Kr++},  
    wherein hyperbolic 2-simplices 
    represent the faces. However, the area of a hyperbolic triangle is, as opposed to the Euclidean one, a function of the {\it angle defect}, i.e., $S(T) = \pi -(\alpha + \beta + \gamma)$, for any triangle $T$ of angles $\alpha,\beta,\gamma$. Therefore, a simple expression as a function of the are cannot be constructed. We suggest the following angle-based expression (see, e.g., \cite{Ja}):
    \begin{equation} \label{eq:HeronHyp}
    w_2(T_{hyp}) = 2\arcsin{\frac{1}{2}\frac{\sqrt{\sinh{p}\sinh(p-w_1)\sinh(p-w_2)\sinh(p-w_1)}}{\cosh\frac{w_1}{2}\cosh\frac{w_2}{2}\cosh\frac{w_3}{2}}}.
    \end{equation}
    Spherical embeddings are used more rarely, but a corresponding formalism can be derived using the following relation from spherical geometry :
    \begin{equation} \label{eq:HeronSph}
    w_2(T_{sph}) = 2\arcsin{\frac{1}{2}\frac{\sqrt{\sin{p}\sin(p-w_1)\sin(p-w_2)\sin(p-w_1)}}{\cos\frac{w_1}{2}\cos\frac{w_2}{2}\cos\frac{w_3}{2}}}.
    \end{equation}
   Clearly, for more general faces, one can use their decomposition into triangles, and the additivity and unicity of the area function. 
   
  For general weights, simpler, combinatorial choice, such as $w_2(T) = w_1+w_2+w_3$ or $w_2(T) = w_1w_2w_3$ , or their scaled version (to dimensionally correspond to area) $w_2(T) = (w_1+w_2+w_3)^2$ (or $w_2(T) = (w_1w_2w_3)^{2/3}$,  respectively) may be employed. An example of such weights can be found in \cite{R++}. In this case, one can not use the decomposition into triangles to define the area of $p$-gons, with $p \geq 4$, but rather use a direct extension of the idea above, e.g for a quadrilateral $Q$, define its weight as $w_2(AQ) = w_1+w_2+w_3, w_4$, etc., where the $w_i$ are defined as before. This type of abstract faces weights (and their higher-dimensional counterparts -- see below) can be trivially generalized to the directed hypernetworks case. 

 The choice of weight $w_m(\tau)$ for an  $m$-dimensional ($m \geq 3$) Euclidean simplex $\tau$ of vertices $v_0,\ldots,v_m$ and edge lengths $d_{ij}; 0 \leq i,j \leq m$ arises naturally from the {\it Cayley-Menger determinant}:
   \begin{equation}                                \label{eq:CM}
   D(\tau) = D(v_0,v_1,\ldots,v_m) = \left| \begin{array}{ccccc}
   0 & 1 & 1 & \ldots & 1 \\
   1 & 0 & d_{01}^{2} & \ldots & d_{0m}^{2} \\
   1 & d_{10}^{2} & 0 & \ldots & d_{1m}^{2} \\
    \cdots & \cdots & \cdots & \cdots & \cdots \\
   1 & d_{m0}^{2} & d_{m1}^{2} & \ldots & 0
   \end{array}
   \right|\;,
   \end{equation}
   given the fact that
   \begin{equation}
   {\rm Vol}^2(\tau) = \frac{-1^{k+1}}{2^k(k!)^2}D(\tau)\,;
   \end{equation}
   (see \cite{Bl}). Thus one can either use $w_m(\tau)  =   {\rm Vol}(\tau)$ or simply put $w_m(\tau)  =   D(\tau)$. 
   This approach extends, like Heron's formula, to spherical and hyperbolic simplices, respectively, (see \cite{Bl}):
   \begin{equation} \label{eq:CM-Sph}
   \Delta(\tau) = \Delta(v_0,\ldots,v_n) = \det{(\cos{d_{ij}})}, 0 \leq i,j \leq m\;.;
   \end{equation}
   %
   %
   \begin{equation}  \label{eq:CM-Hyp}
   \mathcal{D}(\tau) = \mathcal{D}(v_0,\ldots,v_n) = \det{(\cosh{d_{ij}})}; 0 \leq i,j \leq m\;.
   \end{equation}
   %
 Again, one can either use the volume formula from each of the last two formulas, or simply set $w_m(\tau)  =   \Delta(\tau)$ (resp. $w_m(\tau)  =   \mathcal{D}(\tau)$). 
 
 In the case when one does not wish to proceed through the metrization process of the abstract, given weights, one can again introduce combinatorial weights, such as $w_{01} + \ldots + w_{m-1,m}$, or $w_1 \cdots w_{m-1,m}$; where $w_{0,1} \ldots  w_{m-1,m}$ denote the weights of the edges $e(v_0,v_1)$, etc.

   

\subsection{Directionality}
\label{sec:Orientation}
Directionality (or \emph{orientations}) of associations (interactions or commonalities) between elements is an important feature of complex data sets and hence has to be incorporated into network analysis tools. When analyzing higher-order relations, this even more important than in pairwise once, as it determines which complexes are "filled in", i.e., considered in the analysis.

To make this more clear, we start with the basic and most important case of triangles. In this case, there are two possible sets of orientations on the edges: Either the edges define a cycle (closed path), or there is a unique ``source'' and ``sink'' (see Fig.~\ref{Fig:TriangOrient}A). One can distinguish between the two cases by observing that in the acyclic case the information flow is disrupted, whereas in the cyclic case it is not. This approach generalizes readily to higher order $p$-gons (see Fig.~\ref{Fig:TriangOrient} and its accompanying caption).
\begin{figure}[h] \label{Fig:TriangOrient}
	\begin{center}
	\includegraphics[width=\columnwidth]{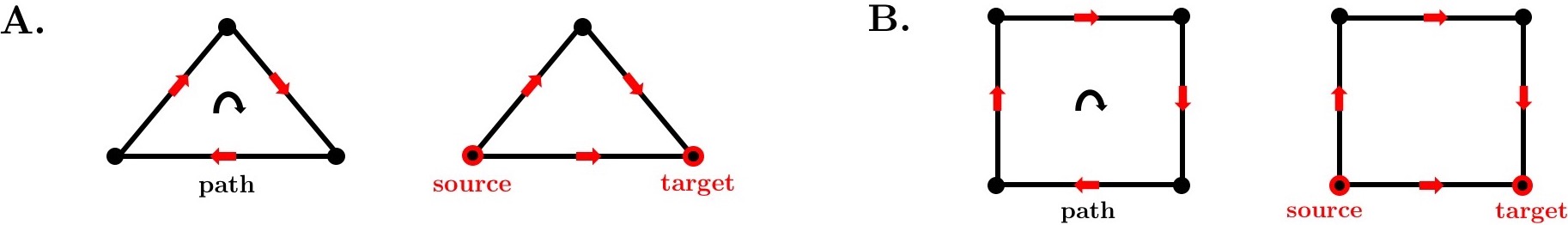}
   \end{center}	
	\caption{Orientations on the edges in \textbf{A} 3-faces (triangles) and \textbf{B} 4-faces (quadrangles). The edges define either a cycle (i.e., a closed path) or an acyclic path with a unique ``source'' and ``sink'' (see Fig.~\ref{Fig:TriangOrient}A).}
\end{figure}
From a purely geometric viewpoint, the ``correct'' choice is that of cycles. It has the advantage of allowing for the standard extension to higher dimensional simplices, e.g. tetrahedra, by filling-in only those with cyclic faces, while requiring, that adjacent faces have opposite orientations. While this process inductively extends to all dimensions, it clearly drastically reduces the number of oriented simplices considered in the analysis. In a practical context, the second option is more intuitive as it relates to information flow in the network (see, e.g. \cite{R++}, \cite{CB}). If one is interested in modeling the information flow in the network, one has to fill in only faces that form acyclic paths with a unique source and target. This is illustrated in Fig.~\ref{Fig:EdgeOrient}. Note, that for computing the Forman-Ricci curvature of a cell, one has to ascertain that its parents and children are uniquely defined (cf. Eq.~\ref{eq:Forman-ch1}). Thus one has to extend the genealogical correlation of nodes/faces to the directed case, see Fig.~\ref{Fig:EdgeOrient} below. 

\begin{figure}[h] \label{Fig:EdgeOrient}
	\begin{center}
	\includegraphics[width=.45\columnwidth]{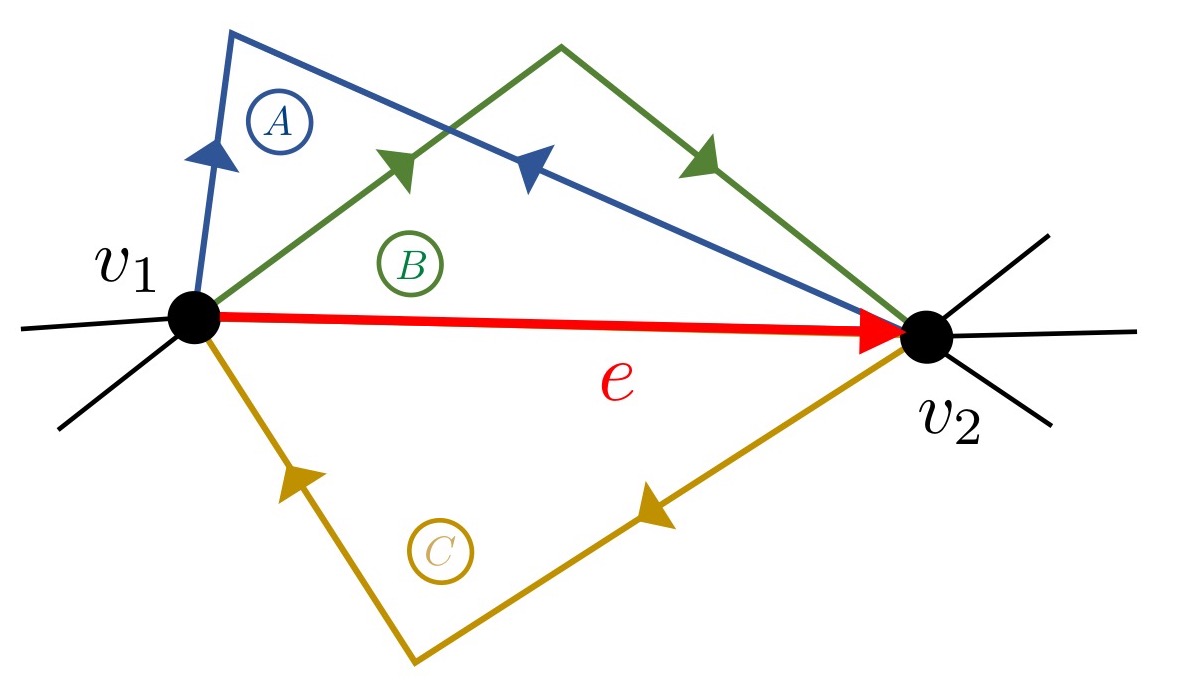}
	\end{center}
	\caption{When considering the information flow within the network, we have to decide which faces to "fill in'', i.e., which oriented triangles are relevant for our analysis. In the context of information flow within the system, one is interested in the parallel routing of information from a source to a target, therefore one should consider only triangles of form (B). However, if one is interested in the collection of information at the source and its spreading at the target, then both ``(B)'' and ``(C)'' triangles should be considered since they contribute to the information flow along the edge $e$. In contrast, triangle ``(A)''  does not contribute to the information flow, so it should be discarded.}
\end{figure}


\section{Discrete Curvatures for Hypernetworks}\label{sec:curvature}
\subsection{Forman-Ricci Curvature}
In this section, we will introduce the Forman Ricci curvature ${\rm Ric}_F(h)$ for hyperedges $h$. The curvature bounds, more so than the actual values, give insight into the structure of the hypernetwork. With the formalism introduced below, we compute global curvature bounds
\begin{align*}
{\rm Ric}^E_m = \min\{{\rm Ric}(E)\,|\, E \in \mathcal{E}\} \quad {\rm and} \quad {\rm Ric}^E_M = \max\{{\rm Ric}(E)\,|\, E \in \mathcal{E}\} \; ,
\end{align*}
which characterize (global) geometric and topological properties of the hypernetwork.

Of course, one could also compute the "graph" Forman curvature ${\bf F}(e)$ (i.e., the 1-dimensional, edge-based notion; see, e.g., ~\cite{WSJ}) for every edge in the complex parametrization and fit lower and upper bounds respectively:
\begin{align*}
{\bf F}^E_m = \min\{{\bf F}(E)\,|\, E \in \mathcal{E}\} \quad {\rm and} \quad
{\bf F}^E_M = \max\{{\bf F}(E)\,|\, E \in \mathcal{E}\} \; .
\end{align*}
These represent bounds for the information flow along the hyperedges. While this approach is computationally much simpler, it loses some of the dimensional information (encoded in the higher-dimensional relationships).

Thus one can employ Forman's Ricci curvature at two levels: Either at a small scale as the Forman curvature of the simplices (polyhedra) that represent the hyperedges of the network, or the "graph" Forman curvature $\mathbf{F}(e)$ of all edge in the network for an analysis of the coarse geometry. The two methods, used in tandem, should yield the most comprehensive understanding of a hypernetwork. 

Recall, that Forman's Ricci curvature for cell complexes is given as follow~\cite{Fo}:
\begin{definition}[Forman-Ricci curvature for weighted cell complexes]
For each $p$-cell $\alpha$ we define
\begin{small}
\begin{equation} \label{eq:Forman-ch1}
\hspace*{-0.75cm}
\mathcal{F}(\alpha^p) = \omega(\alpha^p)\Big[\Big(\sum_{\beta^{p+1}>\alpha^p}\frac{\omega(\alpha^p)}{\omega(\beta^{p+1})}\;
+ \sum_{\gamma^{p-1}<\alpha^p}\frac{w(\gamma^{p-1})}{w(\alpha^p)}\Big)\; 
\end{equation}
\[
\hspace*{3.2cm}
-\sum_{\alpha_1^p\parallel \alpha^p, \alpha_1^p \neq \alpha^p}\Big|\sum_{\substack{\beta^{p+1}>\alpha_1^p \\ \beta^{p+1}>\alpha^p}}\frac{\sqrt{\omega(\alpha^p) \omega(\alpha_1^p)}}{\omega(\beta^{p+1})}
%
- \sum_{\substack{\gamma^{p-1}<\alpha_1^p \\ \gamma^{p-1}<\alpha^p}}\frac{\omega(\gamma^{p-1})}{\sqrt{\omega(\alpha^p) \omega(\alpha_1^p)}}\Big|\:\;\Big]\,;
\]
\end{small}
\noindent where $\alpha < \beta$ indicates, that $\alpha$ is a \emph{parent} of $\beta$ and $\beta$ a \emph{child} of $\alpha$, and
$\alpha_1 \parallel \alpha_2$ that the simplices $\alpha_1$ and $\alpha_2$ are \emph{parallel}. For details, see~\cite{WSJ1} and the references therein.
\end{definition}
Weights are usually chosen to resemble intuitive geometric notions, such as length and volume. In practice, weights can often be constructed to incorporate meta-information. For instance, in a word network built from co-occurrences, the co-occurrence frequencies can be implemented in the edge weights. For the case of hypernetworks, we are interested in computing Forman's curvature for each hyperedge. Given a hyperedge $h$, let us denote the set of \textit{left neighbors} $V_l(h)$ as the set of all vertices in incoming direction with respect to $h$, and $V_r(h)$ as the set of \textit{right neighbors} in the outgoing direction, respectively (see Fig.~\ref{fig:HyperNet}). Then, one can define the Forman-Ricci curvature as the total amount of "information flow" through $h$:
\begin{definition}[Forman's curvature for hyperedges] 
Let $h=\left(	\lbrace	v_r^i \rbrace, \lbrace	v_l^j \rbrace \right)$ be a hyperedge with left neighbors $V_l(h)=\lbrace	v_l^i \rbrace$ and right neighbors $V_r(h)=\lbrace	v_r^j \rbrace$. 
\begin{align*}
{\rm Ric_{F}}(h) = {\rm Ric_{F,V_l(h)}}(h) - {\rm Ric_{F,V_r(h)}}(h)\, ,
\end{align*}
where ${\rm Ric_{F,V_l(h)}}$ measures curvature in the incoming and ${\rm Ric_{F,V_r(h)}}$ in the outgoing direction. For the incoming direction, we compute the Forman-Ricci curvature over the simplex $\mathcal{S} = \left[ 	\lbrace	v_l^i \rbrace, h	\right]$ (see Fig.~\ref{fig:HyperNet}) as
\begin{small}
\begin{align*}
{\rm Ric_{F,V_l(h)}}(h) &= \omega\left(\left[ v_1^1, ... , v_l^{|V_l(h)|}	\right]\right) \left(
\frac{\omega \left(\left[ v_1^1, ... , v_l^{|V_l(h)|} \right]\right)}{\omega \left(\mathcal{S}\right)} + \sum_{\substack{u,v \in V_l(h) \\ u \sim v }} \frac{\omega \left(e=(u,v) \right)}{\omega\left(\left[ v_1^1, ... , v_l^{|V_l(h)|}\right]\right)} \right. \\
	&- \left. \sum_{\substack{u,v \in V_l(h) \\ u \sim v }} \frac{\sqrt{\omega \left(\vartriangle_{u,v,h}\right) \omega\left(\left[ v_1^1, ... , v_l^{| V_l(h) |}	\right] \right) }}{\omega(\mathcal{S})} +   \sum_{\substack{u,v \in V_l(h) \\ u \sim v }}	\frac{\omega(e=(u,v))}{\sqrt{\omega \left(\vartriangle_{u,v,h}\right) \omega\left(\left[ v_1^1, ... , v_l^{|V_l(h)|}\right]\right)}} \right) \; .
\end{align*}
\end{small}
An analog expression can be derived for the outgoing direction.
\end{definition}

\subsection{Ollivier's Ricci Curvature}
We have focused on Forman's Ricci curvature so far, due to the intuitive geometric definition that makes it ideally suited for our parametrization of (hyper)networks as polyhedral complexes. In addition, its computational efficiency makes it ideally suited for large-scale applications.
However, one is naturally interested in extending Ollivier's Ricci curvature (see, e.g.~\cite{Vi}) from networks (where it has become an important tool and subject of research) to hypernetworks.  Unfortunately, generalizing the notion of optimal transport from node to node along edges to higher dimensional faces is far from straightforward. Recently, a new type of curvature, the {\it coarse scalar curvature}, was introduced~\cite{AGE} as a first extension of Ollivier's methods to hypernetworks. We show that, based on the parametrization of hypernetworks as polyhedral complexes, it is possible to extend Ollivier's curvature to hypernetworks in a more general setting. 

This extension is based on a simple, geometric idea: We pass from complex parametrization to its dual complex (see, e.g.~\cite{Hu}). For instance, for each two-dimensional face we can construct a corresponding (dual) node, and two such nodes are connected by an edge if their corresponding faces are adjacent in the complex. The volume (weight) of the original face is then chosen as the weight of the (dual) node. 

The basic idea is quite simple and natural, from a geometric viewpoint (and common in Topological Graph Theory and its applications). Namely, one passes to the dual of the polyhedral complex (network) (see. e.g. \cite{Hu}). For instance, to each top-dimensional face, there corresponds a node, and two such nodes are connected by a (dual) edge if their corresponding faces are adjacent in the given complex/hypernetwork. The measure (e.g. volume) of the original face/simplex is concentrated in the weight of the node. In most instances, the mass transport of interest is between the highest dimensional faces, but this approach can be applied for lower dimensional ones as well, and, after some adaptation, to transport between faces of different dimensions. 

More precisely, if $c^k_1,c^k_2$ are two adjacent $k$-cells of weights $w_1^k,w_2^k$,  let us denote their barycenters, i.e., their \emph{dual} nodes, by $n_1 = n^0_1, n_2 = n_2^0$ and a \emph{dual} edge $(n_1,n_2)$ connecting them by $e_{1,2}$. Furthermore, let $w_{12} = w^1(e_{12})$ be the weight of the edge $e_{12}$. To the dual nodes $n_1, n_2$ we assign the weight of the cells $c_1,c_2$, respectively, i.e. $w_1^0 = w(n^0_1) = w_1^k$ and $w_2^0 = w(n^0_2) = w_2^k$. There exist a number of possible choices for the (dual) edge weights. The simplest is a combinatorial weight of 1, whereas the most natural one from a geometric perspective is $w_{12} = w_{12}^{k-1}$, where $w_{12}^{k-1}$ represents the weight of $f_{12}^{k-1}$ -- the $(k-1)$-dimensional face common to $c_1$ and $c_2$. A further option is that of the Euclidean distance between the barycenters of the original faces in an $\mathbb{R}^N$-embedding of the hypernetwork (see also section \ref{sec:param}).

The (coarse) Ollivier curvature (~\cite{Ollivier}) of the edge $e_{12}$ is then given by 
\begin{equation}
\kappa(e_{12}) = \kappa(n_1,n_2) = 1 - \frac{W_1(n_1,n_2)}{d(n_1,n_2)}\,;
\end{equation}
where $W_1(n_1,n_2)$ represents the Wasserstein distance between the nodes $n_1,n_2$, which in the discrete case at hand is expressed as
\begin{equation}
\label{eq:Ollivier3}
W_1(w'_1, w'_2)=\inf_{\mu_{n_1,n_2}\in \prod(w_1, w_2)}\sum_{(m,n)\in V\times V}d(m, n)\mu_{n_1,n_2}(m,n) \; .
\end{equation}
Here, $\prod(w_1, w_2)$ is the set of probability measures $\mu_{n_1,n_2}$ that satisfy
\begin{equation}
\label{eq:Ollivier4}
\sum_{n \in V}\mu_{n_1,n_2}(m,n)=w_1(m), \,\,\sum_{m\in V}\mu_{n_1,n_2}(m,n)=w_2(n) \; ,
\end{equation}
where $w'_1,w'_2$ represent the {\em normalized} weights corresponding to the vertices $n_1,n_2$. Here $d$ denotes the metric corresponding to the choice of edge weights above.

%


\section{Applications}
\subsection{Applications in Network Analysis}
One of the main advantages of Forman's Ricci curvature resides in the fact that it reveals important algebraic-topological properties of multiplex networks~\cite{SJ}. Moreover, as we have argued above, these benefits extend to hypernetworks, thus allowing the exploration of the structure of such networks, in a manner akin to that of Persistent Homology. More precisely, positive Forman Ricci curvature implies that the first homology group $H_1(\mathcal{N}, \mathbb{R})$ vanishes. In fact, if Forman Ricci curvature is strictly positive, the fundamental group $\pi_1(\mathcal{N})$ is finite. This capability of determining properties of the homotopy, and not just the homology of a hypernetwork, represents an advantage of our approach over Persistent Homology. Note that both properties above generalize, essentially, to any dimension $n$, by passing to the higher-dimensional curvatures generalizing Ricci curvature and introduced by Forman \cite{Fo}.

Yet the benefits of viewing hypernetworks as polyhedral complexes and, in Forman's and Ollivier's Ricci curvature are not restricted to the detection of their core topological structure. Exploring the two curvatures in different dimensions allows for the understanding and classification of various types of networks, especially the differentiation of artificial versus natural hypernetworks. To these advantages, we have already added that of naturally producing structure-preserving embeddings into Euclidean space, which allow for a simpler and more intuitive understanding of these structures.
\subsection{Use case}
We want to give a brief illustration of the geometric tools introduced in section~\ref{sec:curvature}. For this, we computed the Forman-Ricci for a hypernetwork that models historic relationships between object-oriented programming languages (see Fig.~\ref{Fig:languages}). We parametrize the hypernetwork as polyhedral complexes (sec.~\ref{sec:param}). In the unweighted case, the Forman-Ricci curvature reduces to a function of the \emph{in} and \emph{out degree} (i.e., ${\rm Ric}_F (h) = {\rm deg}_{IN} (h) - {\rm deg}_{OUT} (h) \; $). This resembles a combinatorial measure of the information flow through the hyperedge.
\begin{figure}[h] \label{Fig:languages}
	\begin{center}
	\includegraphics[width=0.8\columnwidth]{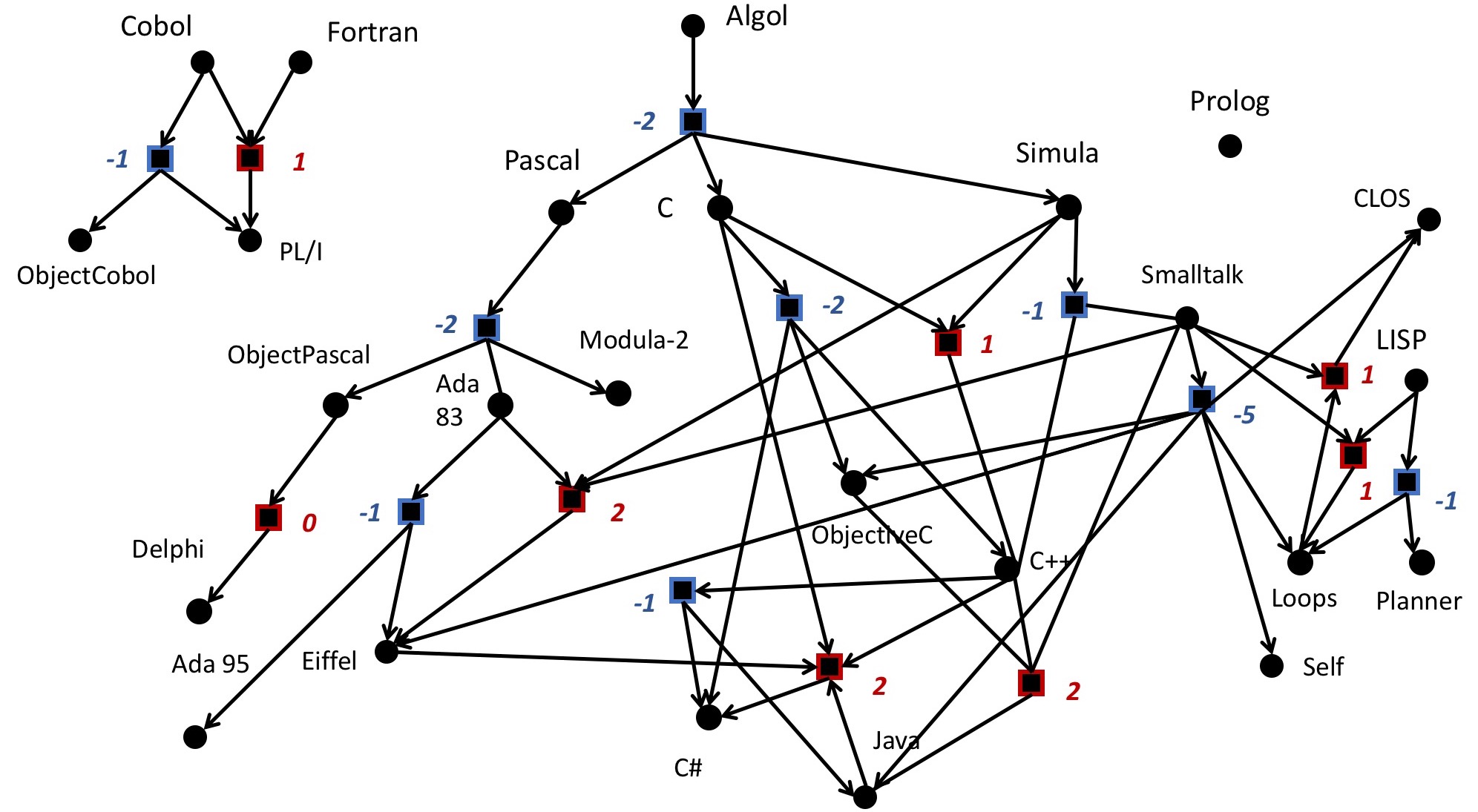}
	\end{center}
	\caption{Family tree of (object-oriented) programming languages as hypernetwork. Hyperedges that represent \emph{co-child} relationships are marked in blue, \emph{co-parent} relationships in red. Each hyperedge is labeled with its respective curvature value.}
\end{figure}

\section{Conclusions}
We introduced a parametrization of hypernetworks as polyhedral complexes that allows for studying higher order relations in complex systems with network-theoretic tools. Furthermore, we discuss discrete notions of curvature in the context of this parametrization and give an explicit formalism for computing these in practice. For the Forman-Ricci curvature, we demonstrate the computation on a small use case.

The present work is merely an exposition of theoretical aspects. Future directions include the empirical explorations of the introduced concepts, in particular, an empirical study of larger datasets from different branches of network science. Moreover, we would like to investigate systematically, how curvature-based analysis of hypernetworks can give insight into the coarse geometry and therefore global structural patterns of such objects. 

For this, we plan to systematically explore different type of weights and faces of varying dimension. This should include a statistical comparison of the results, to ascertain their relative efficiency, especially in comparison with other invariants and network characteristics. We also plan to further investigate the commonalities and differences for the approaches to Ricci curvature that we introduced, i.e. (i) 'graph' Forman curvature and 'complex' Forman curvature and (ii) Forman- and Ollivier-Ricci curvature for hypernetworks. In practice, Forman's curvature will likely be more applicable and scalable, due to its simple notion that avoids computationally expensive tasks, such as the computation of Wasserstein distances. 

%
%
%


\end{document}